\documentclass[conference,letterpaper]{IEEEtran}
\IEEEoverridecommandlockouts

\usepackage[utf8]{inputenc} 
\usepackage[T1]{fontenc}
\usepackage{url}
\usepackage{ifthen}
\usepackage{cite}
\usepackage[cmex10]{amsmath}

\usepackage{amsthm,amssymb,amsfonts}
\usepackage[letterpaper, top=0.75in, bottom=1.1in, left=0.64in, right=0.64in]{geometry}
\usepackage{makecell}
\usepackage [
    n, 
    advantage,
    operators,
    sets,
    adversary,
    landau,
    probability,
    notions,
    logic,
    ff, 
    mm,
    primitives,
    events,
    complexity,
    oracles,
    asymptotics,
    keys
]{cryptocode}

\newtheorem{definition}{Definition}
\newtheorem{theorem}{Theorem}
\newtheorem{lemma}{Lemma}

\begin{document}
\title{On the Context-Hiding Property of Shamir-Based Homomorphic Secret Sharing
}

\author{
Shuai Feng and Liang Feng Zhang
\thanks{Shuai Feng and Liang Feng Zhang are with the School of
 Information Science and Technology, ShanghaiTech University, Shanghai,
 China. Email: {fengshuai,zhanglf}@shanghaitech.edu.cn. This work was
 supported in part by the National Natural Science Foundation of China (No. 62372299).} 
             }
\maketitle

\begin{abstract}
Homomorphic secret sharing (HSS) allows multiple input clients to secretly share their private inputs to a function among several servers such that each server can homomorphically compute the function over its share to produce a share of the function's output. In HSS-enabled applications such as secure multi-party computation (MPC), security requires that the output shares leak no more information about the inputs than the function output. Such security is ensured by the context-hiding property of HSS. The typical rerandomization technique achieves context hiding but increases the share size. To address this, we formalize the context-hiding property of HSS for individual functions, examine the context-hiding property of Shamir-based HSS for monomials, and extend the study to polynomials. 
\end{abstract}

\section{Introduction}
An $n$-input $m$-server homomorphic secret sharing (HSS) scheme \cite{BGI16,ILM21} allows $n$ input clients to secretly share their private inputs $\mathbf x=(x_1,\ldots,x_n)$ among $m$ servers such that each server learns no information about $\bf x$ but can locally compute a function $f$ over its 
share (called \emph{input share}) of the inputs $\mathbf x$ to produce a share (called \emph{output share})
of the function's output $f({\bf x})$, and finally an output client can recover $f(\mathbf x)$ from all output shares. 
An HSS scheme is said to be \emph{information-theoretically (IT) $t$-private} \cite{ILM21,FIKW22} 
if the collusion of any $\le t$ servers learns no information about the inputs $\bf x$.
The communication efficiency of an HSS scheme may be measured by its \emph{upload rate}, which is the ratio between the size of $\bf x$ and that of all input shares, and the 
 {\em download rate}, which is the ratio between the size of $f({\bf x})$ and that of all output shares.
  The existing HSS schemes \cite{ILM21,FIKW22} may enable the evaluation of 
    multivariate polynomials and are based on the  
    well-known secret sharing schemes such as CNF \cite{ISN89} and Shamir \cite{Sha79}. 
In particular, the  Shamir-based HSS is significantly more efficient in
  communication. 

While the existing HSS achieves strong IT security of leaking
 absolutely no information about $\bf x$ to the  servers, 
 it remains a question whether an output client can learn more information about $\mathbf x$ than what is implied by the function output $ f({\bf x})$. 
Ideally, we would like a negative answer to this question, i.e., the output client should
learn  nothing about $\mathbf x$ other than $f({\bf x})$.  
This requirement  is well-motivated in  HSS-based applications such as secure multi-party computation (MPC) \cite{YO19,ILM21} where
 the security may require   a proper subset of the input clients and  $\leq t$ servers learn  no more information about the inputs of the other clients than what follows from their own inputs and  
  outputs.

The nice property of keeping an output client innocent about the inputs has been called
\emph{context hiding} and introduced in the context of homomorphic signatures \cite{SBB19, ALP13} and later extended to HSS \cite{LMS18,CF15}. 
Most existing HSS schemes \cite{CF15,CZL23,LMS18,XW23,PSAM20,EN21} achieve the context-hiding property by “rerandomizing” the output shares in an underlying scheme (i.e., masking the  output shares with shares of 0),
regardless of
 whether the underlying scheme already achieves the context-hiding property for {\em individual} functions or not. 
They require the input clients to  send  shares of 0 to   servers and 
thus resulting in degraded upload rates. 
In particular, if the underlying IT-HSS is Shamir-based, by adding the rerandomization steps {\em the upload rate will be halved}.

The state of the art shows that it is important to conduct a fine-grained
study on the context-hiding property of HSS schemes for   individual
functions, 
because  the study may uncover the fact that the schemes already 
achieve the context-hiding property for specific functions without rerandomization   and 
 thus significantly 
improve the upload rate by removing the  
unnecessary steps. 
Fosli et al. \cite{FIKW22} studied the {\em symmetric privacy}, which is the same notion as context hiding, in HSS schemes,
and examined the property for the monomials $x_1 x_2$ and  $x_1 x_2 x_3$. 
In this paper, we focus on the context-hiding property of Shamir-based IT-HSS for more polynomial functions, and make the following contributions.

\begin{table}[!t]
    \centering
    \caption{Context-hiding property of Shamir-based HSS for Monomials}
    \label{tab:allthms}
    \begin{tabular}{c|c|c|c|c}
    \hline
         Monomial & Domain & \# Variates  & Degree & Context-Hiding \\
       \hline
           $1$ & non-empty set  & $0$  &  $0$  & \checkmark \\
         $x$ & $\mathbb F_p$ & $1$  & $1$  & 
 \checkmark \\
          $x_1 \cdots x_d$ & $\mathbb F_p^d$ & $d $  &  $d\ge 2 $  &  $\times$ \\
          $x^d$ & $\mathbb F_p$ & $1$  & $d \ge 1$ & 
 \checkmark \\
           $x_1^{d_1} \cdots x_n^{d_n}$ & $(\mathbb F_p^*)^n$ & $n \ge 2$  & $d_i\ge 1$ & \checkmark \\
         \hline   
    \end{tabular}
\end{table}

\noindent \textbf{Definition of context-hiding property.} 
We formally define the context-hiding property of HSS for each individual function $f$.
Specifically, we require that the output client should be unable to distinguish between any two different inputs $\mathbf x^{(0)},\mathbf x^{(1)}$ that result in the same function output, i.e., $f(\mathbf x^{(0)}) = f(\mathbf x^{(1)})$.

\noindent \textbf{Context-hiding property of Shamir-based HSS.} 
We study the context-hiding property of the Shamir-based HSS  \cite{FIKW22} for monomials and summarize the results
 in \textsc{TABLE} \ref{tab:allthms}. 
We define the equivalence of polynomials and extend the results for monomials to the equivalent 
polynomials.

\section{Preliminaries}
\label{sec_preliminaries}
Let $\lambda\in\mathbb N$ be a security parameter.
We say that a function $f(\lambda)$ is
   {\em negligible}  in $\lambda$ and denote $f(\lambda)={\sf negl}(\lambda)$, if
   $f(\lambda)=o(\lambda^{-c})$ for any constant $c>0$.
We denote the degree of a polynomial function $f$ by $\mathsf{deg}(f)$, 
the domain of $f$ by ${\cal D}_f$, 
and the range of $f$ by ${\cal R}_f$. 
For any integer $n>0$ and $a<n$, we denote $[n] = \{1, \ldots,n\}$ and $[a,n] = \{a, \ldots,n\}$.
For any prime  $p$, we denote by $\mathbb F_p$ the finite field of order $p$ and denote by $\mathbb{GL}_n(\mathbb F_p)$ the degree-$n$ general linear group over $\mathbb F_p$.
For any finite set $S$, we denote by ``$s\gets S$'' the process of
sampling  an element $s$ uniformly  from $S$.
For any algorithm $\sf Alg$, we denote by ``$y\leftarrow {\sf Alg}(x)$'' the process of running $\sf Alg$ on an input $x$ and assigning its output to $y$.
For any vector $\mathbf x$ of length $n$ and any $i\in[n]$, we denote by $x_i$ the $i$-th entry of $\mathbf x$.

\subsection{Information-Theoretic Homomorphic Secret Sharing}
 
\begin{definition}[Information-Theoretic Homomorphic Secret Sharing (IT-HSS) \cite{ILM21}]
An \emph{$n$-input $m$-server information-theoretic homomorphic secret sharing} scheme $\mathsf{IT\text{-}HSS}=(\mathsf{Share},\mathsf{Eval},\mathsf{Dec})$
for $n$-variate degree-$d$ polynomials over $\mathbb F_p$ consists of the following algorithms:
    \begin{itemize}
        \item $\mathsf{Share}(\mathbf x)\to (\mathbf{s}_j)_{j=1}^m$: Given an input vector $\mathbf{x}=(x_1,\ldots,x_n)$,
        the \emph{sharing} algorithm computes  a tuple 
        of $m$ input shares $(\mathbf{s}_j)_{j=1}^m$.
        \item $\mathsf{Eval}(j,f,\mathbf{s}_j)\to y_j$:
        Given an index $j$, a function $f\in {\cal F}$ and an input share $\mathbf{s}_j$,
        the \emph{evaluation} algorithm executed by a server ${\cal S}_j$ computes  an output share $y_j$.
        \item $\mathsf{Dec}((y_j)_{j=1}^m)\to y$:
        Given the output shares
        $(y_j)_{j=1}^m$, the \emph{decoding} algorithm outputs
         the evaluation result $y$.
    \end{itemize}
\end{definition}

An IT-HSS scheme should satisfy the properties of correctness and privacy. An IT-HSS scheme is \emph{correct} if $\sf Dec$ outputs the correct $f({\bf x})$ when all algorithms are correctly executed.

\begin{definition}[Correctness]
    An IT-HSS scheme is {\emph{correct}} if {for any polynomial $f\in \mathbb F_p[\mathbf x]$ of degree at most $d$,}
    any  $\mathbf{x}\in {\cal D}_f$,
    any $(\mathbf{s}_j)_{j=1}^m\gets \mathsf{Share}(\mathbf{x})$, and any $\{y_j\gets \mathsf{Eval}(j,f,\mathbf{s}_j)\}_{j=1}^m$, 
    it holds that $\mathsf{Dec}((y_j)_{j=1}^m)=f(\mathbf x)$.
\end{definition}

An IT-HSS scheme is {\em $t$-private} if  
 no information about the inputs is disclosed to any $t$ servers.

\begin{definition}[Privacy] 
    An IT-HSS scheme is $t$-\emph{private} if for any   set $T\subseteq [m]$ of cardinality $\le t$ and any two different inputs $\mathbf x^{(0)}$ and $\mathbf x^{(1)}$, $\mathsf{Share}_T(\mathbf x^{(0)})$ and $\mathsf{Share}_T(\mathbf x^{(1)})$ are identically distributed, where the subscript $T$ means restricting to the entries of the output of $\sf Share$ labeled by the indices in $T$. 
\end{definition}

\subsection{Shamir-Based IT-HSS Scheme}

Several recent works \cite{CLZ24,FIKW22,ILM21,ZW22} have used Shamir's threshold scheme \cite{Sha79} 
to build $n$-input $m$-server $t$-private IT-HSS schemes, where every input share consists of
 $n$ field elements and every output share consists of a single field element. 
The basic idea is as follows:  each input client shares an input $x_i~(i\in[n])$ among $m$ servers by executing the share generation procedure of Shamir's $(t,m)$-threshold scheme;
to  evaluate  a degree-$d$ (where $d < m /t $) polynomial $f$ on  $\mathbf x = (x_1,\ldots, x_n)$, each server $\mathcal S_j$ locally computes $f$ on its input share $\mathbf s_j=(s_{j,1},\ldots,s_{j,n})$ to generate
an   output share   $f(\mathbf s_j)$; finally, the output client  
invokes Lagrange interpolation on the output   shares  to recover 
$f(\mathbf x)$.

\begin{definition}[Shamir-Based IT-HSS  \cite{ILM21,FIKW22}]\label{def:ithss}
    An \emph{$n$-input $m$-server $t$-private Shamir-based IT-HSS} scheme $\Pi=(\mathsf{Share},\mathsf{Eval},\mathsf{Dec})$
for $n$-variate degree-$d$ (where $d < m / t$)  polynomials over $\mathbb F_p$ consists of the following algorithms:
    \begin{itemize}
        \item $\mathsf{Share}(\mathbf x)\to (\mathbf{s}_j)_{j=1}^m$: Choose a random polynomial $\varphi\in\mathbb F_p^n[u]$ of degree $\le t$ such that $\varphi(0)=\mathbf x$; compute the $j$-th server's input share $\mathbf s_j=\varphi(j)$ for every $j\in[m]$.
        \item $\mathsf{Eval}(j,f,\mathbf{s}_j)\to y_j$:
        Server $\mathcal S_j$ computes an output share $y_j=f(\mathbf s_j)$.
        \item $\mathsf{Dec}((y_j)_{j=1}^m)\to y$:
        Since $f\circ \varphi$ is a polynomial of degree $\le dt$, invoke Lagrange interpolation to recover 
         the evaluation result as
$$
y = \sum_{j=1}^m y_j \prod_{k\in[m],k\ne j}\frac{k}{k-j}.
$$
    \end{itemize}
\end{definition}

Correctness of Shamir-based IT-HSS has been proven in \cite{ILM21} and  $t$-privacy follows from Shamir's scheme \cite{Sha79}.

\section{Defining Context-Hiding Property of IT-HSS}
In IT-HSS, the output client may be different from the input clients and try to learn information about the inputs $\bf x$ from the received output shares. In some applications such as MPC, the input clients are reluctant to leak $\bf x$ to the output client. To prevent the output client from learning $\bf x$, the context-hiding property has been considered in \cite{LMS18,CF15} and requires that the output client cannot distinguish between the output shares generated by the servers and the simulated ones. 

In \cite{FIKW22}, an IT-HSS scheme is said to be symmetrically private (i.e., context-hiding) if {\em for any function} $f$ from a function family ${\cal F}$ under consideration, the joint distribution of all output shares depends only on $f(\mathbf x)$. 
Specifically, it requires that for any two different inputs ${\bf x}^{(0)},{\bf x}^{(1)}$ that result in the same function output $y$,  the corresponding output shares 
must be {\em identically distributed}. 
Such a definition is somewhat one-sided. 
On the one hand, when the corresponding output shares are {\em statistically indistinguishable}, it is also difficult for an adversary to distinguish between the two inputs. 
On the other hand, if the scheme fails to meet the context-hiding requirement for just {\em a single function}
in ${\cal F}$, the scheme fails to be context-hiding.

To study context-hiding property in a more fine-grained way
we 
 introduce a {\em context hiding experiment} $\mathbf{Exp}_{\mathcal A,\mathsf{IT}\text{-}\mathsf{HSS}}^{\mathsf{Ctx\text{-}Hid}}(f)$ for any individual  function $f$: 
\begin{itemize}
    \item $(\mathbf{x}^{(0)},\mathbf{x}^{(1)})\gets \mathcal A$, s.t., $\mathbf{x}^{(0)}\neq \mathbf{x}^{(1)}$, $f(\mathbf{x}^{(0)})=f(\mathbf{x}^{(1)})$;

    \item
           $ b\gets \{0,1\},
             (\mathbf{s}_j)_{j=1}^m \gets \mathsf{Share}(\mathbf{x}^{(b)})$; 

    \item
       $y_j\gets{\sf Eval}(j, f,\mathbf s_j)$ for all $j\in[m]$;
        
    \item
            $b'\gets \mathcal A((y_j)_{j\in [m]})$;

    \item
            if $b'=b$, output $1$;
            otherwise, output $0$.
\end{itemize}
In the experiment, an adversary $\mathcal A$ declares two different input vectors ${\mathbf x}^{(0)},{\mathbf x}^{(1)}\in {\cal D}_f$ such that $f({\mathbf x}^{(0)})=f({\mathbf x}^{(1)})$. 
The challenger chooses an input vector ${\mathbf x}^{(b)}$ uniformly and runs $\sf Share$ and $\sf Eval$ to generate $m$ output shares. Finally, $\cal A$ observes all output shares and guesses which input vector was used by the challenger.

\begin{definition}[Context Hiding] \label{def_ch}
     Let $\mathcal A$ be an adversary.  
     For any $\lambda\in\mathbb N$, in an IT-HSS scheme $\mathsf{IT\text{-}HSS}$ for a function $f$ over
      a finite field $\mathbb F_p$ of order $p= O(2^\lambda)$,  define
     $\mathcal A$'s {\em advantage} in the context hiding experiment $\mathbf{Exp}_{\mathcal A,\mathsf{IT}\text{-}\mathsf{HSS}}^{\mathsf{Ctx\text{-}Hid}}(f)$   as
    $
    \mathbf{Adv}_{\mathcal A,\mathsf{IT}\text{-}\mathsf{HSS}}^{\sf Ctx\text{-}Hid}(f)= \big|\Pr[\mathbf{Exp}_{\mathcal A,\mathsf{IT}\text{-}\mathsf{HSS}}^{\sf Ctx\text{-}Hid}(f)=1]-1/2\big|
    $,
    where the probability is taken over all randomness in the experiment. 
The scheme $\sf IT$-$\sf HSS$ is called \emph{context-hiding for $f$} if for any adversary $\mathcal A$,  
    $
    \mathbf{Adv}_{\mathcal A,\mathsf{IT}\text{-}\mathsf{HSS}}^{\sf Ctx\text{-}Hid}(f)
    \le\mathsf{negl}(\lambda)
    $.
    The scheme
is {\em perfectly context-hiding for $f$} if for any adversary $\mathcal A$, $\mathbf{Adv}_{\mathcal A,\mathsf{IT}\text{-}\mathsf{HSS}}^{\mathsf{Ctx}\text{-}\mathsf{Hid}}(f)=0$. 
\end{definition}

\noindent
{\em Remark 1.} When studying the context-hiding property for a function $f$, there is no need to 
consider a value $y\in {\cal R}_f$ that has a unique preimage ${\bf x}\in {\cal D}_f$, because  
the output shares simply cannot leak more information than what is implied by $y$, which is known to the output 
client. In particular, a bijective function is always perfectly context-hiding.

\section{Context-Hiding Property of Shamir-Based IT-HSS for Monomials} \label{sec_mon}
In this section, we investigate the context-hiding property of the Shamir-based IT-HSS scheme for monomials. We start from the multilinear monomial $x_1 x_2\cdots x_d$, and then study the general monomial $x_1 ^{d_1} x_2^{d_2}\cdots x_n^{d_n}$. In particular, to compute a degree-$d$ polynomial, the Shamir-based IT-HSS scheme with $t$-privacy requires at least $dt + 1$ servers, i.e., $m\ge dt + 1$.

\subsection{For Multilinear Monomial $x_1 x_2\cdots x_d$}

When the degree of monomial $f$ is $d=0$, $f=1$ is a constant function, then there is no input to the function, so the output shares are independent of the client's input; when the degree of monomial $f$ is $d=1$, the value of $x$ is the information implied in $f(x)$. Thus, we conclude that the Shamir-based IT-HSS scheme is perfectly context-hiding for $f_0= 1$ and $f_1(x) = x$ with ${\cal D}_{f_0}$ being an arbitrary non-empty set and ${\cal D}_{f_1}=\mathbb F_p$. 
\begin{theorem}
    The $m$-server $t$-private Shamir-based IT-HSS scheme $\Pi$ is perfect context-hiding for $f_0 = 1$ and $f_1(x) = x$ with ${\cal D}_{f_0}$ being an arbitrary non-empty set and ${\cal D}_{f_1}=\mathbb F_p$.
\end{theorem}

Next we study the context-hiding property of the Shamir-based IT-HSS scheme for $f(\mathbf x) = \prod_{i=1}^d x_i$ where $d\ge 2$. 

\begin{theorem} \label{thm_m22}
   The $m$-server $t$-private Shamir-based IT-HSS scheme $\Pi$ is not context-hiding for $f(\mathbf x) = \prod_{i=1}^d x_i$ with $d\ge 2$ and ${\cal D}_f=\mathbb F_p^d$.
\end{theorem}
\begin{proof}
The output share from server $\mathcal S_j$ for each $j\in [m]$ is 
$$
    y_j
    =  \prod_{i=1}^d \bigg(x_i + \sum_{u=1}^t j^u r_{u,i}\bigg), 
$$
where $r_{u,i}$ is chosen uniformly from $\mathbb F_p$ for every $u\in [t]$ and $i\in[d]$. 
The adversary $\mathcal A$ chooses two inputs
$$\mathbf x^{(0)} = (0,\ldots,0)\in \mathbb F_p^d, 
\mathbf x^{(1)} = (0,\ldots,0,1)\in \mathbb F_p^d
$$
such that $f(\mathbf x^{(0)}) = f(\mathbf x^{(1)}) = 0$.

If the challenger chooses $b=0$, then the output share from $\mathcal S_j$ for each $j\in [m]$ is 
\begin{align*}
   y^{(0)}_j
   =   \prod_{i=1}^d \bigg(\sum_{u=1}^t j^u r_{u,i}^{(0)} \bigg) 
    =   \sum_{v=d}^{dt} j ^ v 
        \bigg( \sum_{\substack{u_1,\ldots,u_d\in [t]\\u_1 + \cdots + u_d = v}} \prod_{i=1}^d r_{u_i,i}^{(0)}\bigg).
    \end{align*}
If the challenger chooses $b=1$, then the output share from $\mathcal S_j$ for each $j\in [m]$ is 
\begin{align*}
    y^{(1)}_j
    = & \bigg(\prod_{i=1}^{d-1} \bigg(\sum_{u=1}^t j^u r_{u,i} ^{(1)} \bigg)\bigg)\cdot \bigg(1 + \sum_{u=1}^t j^u r_{u,d} ^{(1)} \bigg) \\
     = & \prod_{i=1}^{d-1} \bigg(\sum_{u=1}^t j^u r_{u,i} ^{(1)} \bigg) + \prod_{i=1}^d \bigg(\sum_{u=1}^t j^u r_{u,i} ^{(1)} \bigg) \\
     = & j^{d-1} \prod_{i=1}^{d-1} r_{1,i} ^{(1)} + \sum_{v = d} ^ {(d-1)t} j ^ v 
    {\bigg( \sum_{\substack{u_1,\ldots,u_{d - 1}\in [t]\\u_1 + \cdots + u_{d - 1} = v}} \prod_{i = 1}^{d - 1} r_{u_i,i} ^{(1)} \bigg) }\\
     & + \sum_{w = d}^{dt} j ^ w {\bigg( \sum_{\substack{u_1, \ldots, u_d\in [t] \\ u_1 + \cdots + u_d = w}} \prod_{i = 1}^d r_{u_i, i} ^{(1)}  \bigg) }.
     \\
     = & j^{d - 1} 
          \prod_{i = 1}^{d - 1} r_{1,i} ^{(1)} 
          +  
     \sum_{v = d} ^ {(d - 1) t}  j ^ v 
          \sum_{\substack{u_1, \ldots, u_{d - 1} \in [t],u_d\in [0,t] \\ u_1 + \cdots + u_{d} = v}} \prod_{i = 1}^d r_{u_i, i} ^{(1)}  
          \\
      & + \sum_{w = (d - 1) t + 1}^{dt} j ^ w {
            \bigg(\sum_{\substack{u_1,\ldots,u_d\in [t]\\u_1 + \cdots + u_d = w}} \prod_{i=1}^d r_{u_i,i} ^{(1)}\bigg)},
      \end{align*}
where $r_{0,i}^{(1)}=1$ for all $i\in [d]$.

Note that the $(dt - d + 2) \times (dt - d + 2)$ square matrix
    \begin{equation*}
       A=\left[
        \begin{array}{cccc}
           1  & 2^{d-1} & \cdots & (dt - d + 2)^{d-1} \\
           1  & 2^{d} & \cdots & (dt - d + 2)^d \\
            \vdots & \vdots & \ddots & \vdots\\
           1 & 2^{dt} & \cdots & (dt - d + 2)^{dt}
        \end{array}
        \right]
    \end{equation*}
    has full rank, 
    so there exists a vector $\mathbf a=(a_1,\ldots,a_{dt - d + 2})^\top$ such that 
    $$
    A \mathbf a=(1,0,\cdots,0)^\top\in \mathbb F_p^{dt - d + 2}.
    $$
    Due to this observation, the adversary $\mathcal A$ can perform a linear combination $\sum_{j=1}^{dt - d + 2} a_j y^{(b)}_j$ of the first $(dt - d + 2)$ servers' output shares $\{y^{(b)}_j\}_{j=1}^{dt - d + 2}$ using vector $\mathbf a$ such that 
    $$
    \sum_{j=1}^{dt - d + 2} a_j y^{(0)}_j=0, \sum_{j=1}^{dt - d + 2} a_j y^{(1)}_j=\prod_{i = 1}^{d - 1} r_{1,i} ^{(1)}. 
    $$
    The adversary $\mathcal A$  may guess
    
    \begin{equation*}
        b'=
        \begin{cases}
            0 & \text{if }  \sum_{j=1}^{dt - d + 2} a_j y^{(b)}_j = 0, \\
            1 & \text{otherwise}.
        \end{cases}
    \end{equation*}

    We shall analyze $\mathcal A$'s advantage. First, it is easy to see that 
    $
        \Pr[b' = 0 | b = 0 ] = 1.
    $
    Since $r_{1,i}^{(1)}$ is chosen uniformly from $\mathbb F_p$ for each $i\in [d-1]$, 
    $$
        \Pr[b' = 1 | b = 1 ] = \Pr\left[\prod_{i = 1}^{d - 1} r_{1,i} ^{(1)} \ne 0\right] = \frac{(p-1) ^ {d-1}}{p ^ {d - 1}}.
    $$
    Then 
    the advantage of $\mathcal A$ is computed as
    \begin{align*}
    & \mathbf{Adv}_{\mathcal A,\Pi}^{\mathsf{Ctx}\text{-}\mathsf{Hid}}(f) \\
         = &\big|\Pr[\mathbf{Exp}_{\mathcal A,\Pi}^{\sf Ctx\text{-}Hid}(f)=1]-1/2\big| \\
        = &\big| 1/2\cdot\Pr[b' = 0 | b = 0 ] + 1/2\cdot \Pr[b' = 1 | b = 1 ] -1/2 \big|\\
        = & \frac{(p-1) ^ {d-1}}{2 p ^ {d - 1}},
    \end{align*}
 is not negligible in $\lambda$. Thus, $\Pi$ is not context-hiding for $f(\mathbf x) = \prod_{i=1}^d x_i$ with $d\ge 2$ and ${\cal D}_f=\mathbb F_p^d$. 
\end{proof}

\subsection{For General Monomial $x_1^{d_1} x_2^{d_2}\cdots x_n^{d_n}$}
For general monomial $f(\mathbf x) = \prod_{i=1}^n x_i^{d_i}$ with $d_i\ge 1$ for each $i \in [n]$, we first show that when $n=1$ and $d\ge 1$, Shamir-based IT-HSS is perfectly context-hiding for $f(x) = x^d$ with domain ${\cal D}_f=\mathbb F_p$. 
\begin{theorem} \label{thm_m12}
    The $m$-server $t$-private Shamir-based IT-HSS scheme $\Pi$ is perfectly context-hiding for $f(x) = x^d$ with $d\ge 1$ and ${\cal D}_f=\mathbb F_p$. 
\end{theorem}
\begin{proof} 
Let $\ell=\mathsf{gcd}(d,p-1)$. We consider two cases. 
 
 \noindent \underline{\textsc{Case 1}: $\ell=1$}. In this case, the function $f:\mathbb{F}_p\rightarrow 
\mathbb{F}_p$ is bijective.  By   {\em Remark 1}, 
 the Shamir-based IT-HSS scheme is  perfectly context-hiding
  for $f$.

\noindent \underline{\textsc{Case 2}: $\ell>1$}.
As $y=0$ has a unique image under $f$, by {\em Remark 1} it suffices to consider 
 the function outputs   $y\in {\cal R}_f \setminus \{0\}$.     
Note that   such a $y$ has exactly
$\ell$ different preimages and for any two different preimages $x^{(0)}$ and $x^{(1)}$, there is a constant  $c\in\mathbb F_p^*$ such that $x^{(1)} = c x^{(0)}$. Then
    $$
    (x^{(0)})^d =y= (x^{(1)})^d = (c x^{(0)})^d =c^d (x^{(0)})^d.
    $$
It follows that  $c^d = 1$.
    Consider the experiment $\mathbf{Exp}_{\mathcal A,\Pi}^{\mathsf{Ctx\text{-}Hid}}(f)$  where 
$\cal A$ declares  $x^{(0)}$ and $x^{(1)}$. 
    If the challenger chooses $b\in\{0,1\}$, then for each $j \in [m]$  the output share of
$\mathcal S_j$ is
$$
y_j^{(b)} =  \bigg(x^{(b)} + \sum_{u=1}^t j^u r_{u}^{(b)}\bigg)^d,
$$
where every $r_{u}^{(b)}$ is uniformly chosen  from $\mathbb F_p$.
Below we show that $(y_j^{(0)})_{j=1}^m$ and $(y_j^{(1)})_{j=1}^m$ are identically distributed.
The proof strategy is  finding a bijection
$
\tau((r_{u}^{(0)})_{u \in [t]})=(r_{u}^{(1)})_{u \in [t]}
$
such that replacing $(r_{u}^{(1)})_{u \in [t]}$ by $\tau((r_{u}^{(0)})_{u \in [t]})$ in the expression of $(y_j^{(1)})_{j=1}^m$ gives $(y_j^{(0)})_{j=1}^m$. 
Let $\tau:\mathbb{F}_p^t\rightarrow \mathbb{F}_p^t$ be a bijection defined by 
$\tau(\lambda_1,\ldots,\lambda_t)=c\cdot (\lambda_1,\ldots,\lambda_t) $.
After setting $(r_{u}^{(1)})_{u \in [t]}=\tau((r_{u}^{(0)})_{u \in [t]})$,  we have 
\begin{align*}
\bigg(x^{(1)} + \sum_{u=1}^t j^u r_{u}^{(1)}\bigg)^d  = & \bigg(c x^{(0)} + \sum_{u=1}^t j^u (cr_{u}^{(0)} ) \bigg)^d   
    \\
    = &c^d \bigg( x^{(0)} + \sum_{u=1}^t j^u r_{u}^{(0)}  \bigg)^d   
    \\
    = & \bigg(x^{(0)} + \sum_{u=1}^t j^u r_{u}^{(0)}\bigg)^d.
\end{align*}
Thus, we have that $(y_j^{(1)})_{j=1}^m=(y_j^{(0)})_{j=1}^m$ after pairing the randomness that underlies the computations of the two output share vectors. 
Hence, the two output share vectors are identically distributed. Thus, 
$
 \mathbf{Adv}_{\mathcal A,\Pi}^{\sf Ctx\text{-}Hid}(f)
    =0
$.
\end{proof}
 
While the Shamir-based IT-HSS scheme may not be perfectly context-hiding for a general  
monomial function $f(\mathbf x)=\prod_{i=1}^n x_i^{d_i}~(n\geq 2, d_1,\ldots,d_n\geq 1)$ with domain $\mathbb{F}_p^n$, the following theorem shows that the scheme becomes perfectly context-hiding if we 
restrict $f$ to a smaller domain $(\mathbb F_p^*)^n$.  
\begin{theorem} \label{thm_ms22}
    The $m$-server $t$-private Shamir-based IT-HSS scheme $\Pi$ is perfectly context-hiding for $f(\mathbf x)=\prod_{i=1}^n x_i^{d_i}$, where   $d_i \ge 1$ for every $i\in[n]$ and ${\cal D}_f=(\mathbb F_p^*)^n$.
\end{theorem}

\begin{proof}
Let  $\mathbf x^{(0)}=(x_i^{(0)})_{i=1}^n,\mathbf x^{(1)} = (x_i^{(1)})_{i=1}^n\in (\mathbb F_p^*)^n$
be any two different input vectors such that $f(\mathbf x^{(0)}) = f(\mathbf x^{(1)}) $. 
For every $i\in[d]$, let  $c_i\in \mathbb F_p^*$ be a scalar such that 
 $x_i^{(1)}=c_i x_i^{(0)}$. Then 
    \begin{align*}
           \prod_{i=1}^n (x_i^{(0)})^{d_i} = & \prod_{i=1}^n (x_i^{(1)})^{d_i} \\
           = & \prod_{i=1}^n (c_i x_i^{(0)})^{d_i}  
           =  \bigg( \prod_{i = 1}^n c_i^{d_i} \bigg) \cdot \prod_{i = 1}^n (x_i^{(0)})^{d_i}.
    \end{align*}
Thus, 
        $ \prod_{i = 1}^n c_i^{d_i}  = 1 $.
          Consider the experiment $\mathbf{Exp}_{\mathcal A,\Pi}^{\mathsf{Ctx\text{-}Hid}}(f)$  where 
$\cal A$ declares  ${\bf x}^{(0)}$ and ${\bf x}^{(1)}$. 
     If the challenger chooses $b\in\{0,1\}$, the output share from
$\mathcal S_j$ is 
$$
y_j^{(b)}
    =  \prod_{i=1}^n \bigg(x_i^{(b)} + \sum_{u=1}^t j^u r_{u,i}^{(b)}\bigg)^{d_i}$$
     for all $j \in [m]$,     
where each $r_{u,i}^{(b)}$ is uniformly chosen from $\mathbb F_p$.
Suppose that 
$
    r_{u,i}^{(1)}=c_i r_{u,i}^{(0)}
$
for all $i\in[n], u\in [t]$. Then  
\begin{align*}
 y_j^{(1)}  =  & \prod_{i=1}^n \bigg(x_i^{(1)} + \sum_{u=1}^t j^u r_{u,i}^{(1)}\bigg)^{d_i} \\
      = &\prod_{i=1}^n \bigg(c_i x_i^{(0)} + \sum_{u=1}^t j^u c_i r_{u,i}^{(0)}\bigg)^{d_i}
    \\
       =&\bigg(\prod_{i=1}^n c_i^{d_i} \bigg)\cdot  \prod_{i=1}^n \bigg(x_i^{(0)} + \sum_{u=1}^t j^u r_{u,i}^{(0)}\bigg)^{d_i}\\
    = & \prod_{i=1}^n \bigg(x_i^{(0)} + \sum_{u=1}^t j^u r_{u,i}^{(0)}\bigg)^{d_i}
    =y_j^{(0)}.
\end{align*}
Therefore, we have that $(y_j^{(1)})_{j=1}^m=(y_j^{(0)})_{j=1}^m$ after pairing the randomness
that underlies the computations of the two output share vectors. 
Hence, the two output share vectors are identically distributed. 
Thus, 
$
 \mathbf{Adv}_{\mathcal A,\Pi}^{\sf Ctx\text{-}Hid}(f)
    =0
$.
\end{proof}

\section{Context-Hiding Property of Shamir-Based IT-HSS for Polynomials}
In this section, we extend the   results  for monomials (Section \ref{sec_mon}) to 
a wide range of equivalent polynomials.

\begin{definition}[\bf Equivalence of Polynomials]
Let $f,g\in \mathbb F_p[{\bf x}]$ be $n$-variate degree-$d$ polynomials with domain
${\cal D}_{f},{\cal D}_g\subseteq \mathbb{F}_p^n$ respectively.
We say that $f$ and $g$ are 
     \emph{equivalent} and denote  $(f,{\cal D}_f)\equiv (g,{\cal D}_g)$ if there is an 
     {\em equivalence  transformation} $S=(\alpha,\beta, \gamma, 
     L,{\bf c}, {\bf e})\in \mathbb F_p^*\times \mathbb F_p\times \mathbb F_p \times \mathbb{GL}_n(\mathbb F_p) \times 
     \mathbb F_p^n\times \mathbb F_p^n$ that satisfies the following properties: 
     \begin{align}
  {\cal D}_f&= \{ ({\bf x}+{\bf c}) L + {\bf e}: {\bf x}\in {\cal D}_g \} \label{eqn:ep1}, \\ 
g(\mathbf x)&=\alpha \Big( f\big(({\bf x}+{\bf c}) L + {\bf e}\big) + \beta \Big)+\gamma, \forall {\bf x}\in {\cal D}_g. \label{eqn:ep2}
     \end{align}
\end{definition}

\begin{lemma} \label{prop_eq}
The binary relation ``$\equiv$'' between polynomials is an equivalence relation, i.e., it satisfies the following properties:  
\begin{itemize}
\item[1)] {\em reflexive}: $(f,{\cal D}_f)\equiv (f,{\cal D}_f)$ for any $n$-variate polynomial function $f\in \mathbb{F}_{p}[{\bf x}]$ with domain 
${\cal D}_f\subseteq \mathbb{F}_p^n$;  
\item[2)] {\em symmetric}: If $(f,{\cal D}_f)\equiv (g,{\cal D}_g)$, then $(g,{\cal D}_g)\equiv (f,{\cal D}_f)$;  
\item[3)] {\em transitive}: If $(f,{\cal D}_f)\equiv (g,{\cal D}_g)$ and $(g,{\cal D}_g)\equiv (h,{\cal D}_h)$,
then $(f,{\cal D}_f)\equiv (h,{\cal D}_h)$.
\end{itemize}
\end{lemma}
\begin{proof}
1) reflexive: 
$S=(1,0, 0, 
     I,{\bf 0}, {\bf 0})$ is an equivalence transformation from
$(f, {\cal D}_f)$ to itself, where $I$ is the $n\times n$ identity matrix.
2) symmetric: If $S=(\alpha,\beta,\gamma,
     L, {\bf c},  {\bf e})$ is an equivalence transformation from
 $(f, {\cal D}_f)$  to $(g, {\cal D}_g)$, then   
$S'=(\alpha^{-1},-\gamma,-\beta, 
     L^{-1},-{\bf e}, -{\bf c})$ 
     is an equivalence transformation from $(g, {\cal D}_g)$  to $(f, {\cal D}_f)$.
     3) transitive: If   $S_1=(\alpha_1,\beta_1,\gamma_1,
     L_1, {\bf c}_1,  {\bf e}_1)$ is an equivalence transformation from
 $(f, {\cal D}_f)$  to $(g, {\cal D}_g)$ and 
 $S_2=(\alpha_2,\beta_2,\gamma_2,
     L_2, {\bf c}_2,  {\bf e}_2)$ is an equivalence transformation from
 $(g, {\cal D}_g)$  to $(h, {\cal D}_h)$, then 
 $S_3=(\alpha_2\alpha_1,\beta_1, \alpha_2\gamma_1+\alpha_2\beta_2+\gamma_2,L_2L_1, {\bf c}_2, ({\bf e}_2+{\bf c}_1)L_1+{\bf e}_1)$ is an equivalence transformation from 
 $(f, {\cal D}_f)$  to $(h, {\cal D}_h)$.  
\end{proof}

The following theorem shows that the Shamir-based IT-HSS scheme always exhibits the same 
context-hiding property for any two equivalent polynomials.
\begin{theorem} \label{thm_equiv}
Suppose that  $(f,{\cal D}_f)\equiv (g,{\cal D}_g)$ is an equivalence between two $n$-variate polynomials. The $m$-server $t$-private Shamir-based IT-HSS scheme $\Pi$ is (perfectly) context-hiding for $f$ if and only if it is (perfectly) context-hiding for $g$.
\end{theorem}

\begin{proof} 
We focus on context hiding. The proof for perfect context hiding is similar. 
By 2) of \textbf{Lemma} \ref{prop_eq}, it suffices to show  ``only if''.
Suppose that $\Pi$ is context-hiding for $(f,{\cal D}_f)$. We  need to prove    $\mathbf{Adv}_{\mathcal A,\Pi}^{\sf Ctx\text{-}Hid}(g)\le \mathsf{negl}(\lambda)$ for any adversary $\mathcal A$. It suffices to show for any two different inputs $\mathbf x^{(0)},\mathbf x^{(1)}\in {\cal D}_g$ with $g(\mathbf x^{(0)})=g(\mathbf x^{(1)})$, 
    the  two output share distributions
\begin{align}\label{eqn:og}
    \Big(y_j^{(b)}\Big)_{j=1}^m  = \Big (g\big(\mathbf x^{(b)}+\sum_{u=1}^t j^u \mathbf r_u^{(b)}\big)\Big)_{j=1}^m,~~ b=0,1
\end{align}
    are   statistically indistinguishable, 
        where $\mathbf r_u^{(b)}$ are chosen uniformly 
        from $\mathbb F_p^n$  for all $b\in \{0,1\}$ and $u\in[t]$. 

    Since  $(f,{\cal D}_f)\equiv (g,{\cal D}_g)$,  there is an equivalence 
     transformation $S=(\alpha,\beta, \gamma, 
     L,{\bf c}, {\bf e})$ such that Eq. (\ref{eqn:ep1}) and   (\ref{eqn:ep2}) are satisfied. 
Let  $\bar{\mathbf x}^{(b)}=({\bf x}^{(b)}+{\bf c}) L + {\bf e}$ for  $b\in \{0,1\}$. Then  
 \begin{equation*}
 \begin{split}\label{eq_eqi}
     f(\bar{\mathbf x}^{(0)}) &=   f(({\bf x}^{(0)}+{\bf c}) L + {\bf e})\\ 
    &= \alpha^{-1} ( g(\mathbf x^{(0)}) - \gamma) -\beta \\
     &= \alpha^{-1} ( g(\mathbf x^{(1)}) - \gamma)-\beta \\
     &= f(({\bf x}^{(1)}+{\bf c}) L + {\bf e}) = f(\bar{\mathbf x}^{(1)}).
     \end{split}
 \end{equation*}
Since $\Pi$ is context-hiding  for $f$, the output shares 
$$
         \Big(\bar y_j^{(b)}\Big)_{j=1}^m =  
        \Big( f \big(\bar{\mathbf x}^{(b)}+\sum_{u=1}^t j^u \bar{\mathbf r}_u^{(b)} \big) \Big)_{j=1}^m, ~~b=0,1 
$$ 
must be statistically indistinguishable, 
where $\bar{\mathbf r}_u^{(b)}$ is uniformly  chosen  from $\mathbb F_p^n$  for all $b\in\{0,1\}$ and $u\in[t]$.
If  ${\mathbf r}_u^{(b)} = \bar{\mathbf r}_u^{(b)} L^{-1}$ for all $u\in [t]$, then  
    \begin{equation*}
    \begin{split} \label{eq_equiv}
      y_j^{(b)} = & g\big(\mathbf x^{(b)}+\sum_{u=1}^t j^u \mathbf r_u^{(b)}\big) \\
      =& \alpha \bigg( f\Big(
      \big(
      \mathbf x^{(b)}+\sum_{u=1}^t j^u \mathbf r_u^{(b)} + \mathbf c
      \big) L +{\bf e}  \Big) + \beta \bigg) +\gamma \\
        =& \alpha \Big( f\big(
      \bar{\mathbf x}^{(b)} +  \sum_{u=1}^t j^u \bar{\mathbf r}_u^{(b)} 
     \big) + \beta \Big) +\gamma \\
        = & \alpha (\bar y_j^{(b)} + \beta) + \gamma. 
    \end{split}  
    \end{equation*} 
    Thus,   for any ${\bf z} \in \mathbb F_p^m$, we have that 
    \begin{align*}
     \Pr\left[\big(y_j^{(b)}\big)_{j=1}^m ={\bf z} \right ]=   \Pr\Big[\big(\bar y_j^{(b)}\big)_{j=1}^m = \frac{{\bf z}-\gamma \cdot {\bf 1}}{\alpha}-\beta \cdot {\bf 1}\Big],
    \end{align*}
    where $\bf 1$ is the all-one vector of length $m$.
    This implies that the distributions in (\ref{eqn:og}) are also statistically indistinguishable.
        \end{proof}

According to \textbf{Theorem} \ref{thm_equiv}, we can easily extend the context-hiding property of the Shamir-based IT-HSS scheme $\Pi$ for monomials (Section \ref{sec_mon}) to polynomials equivalent to them. 
For example, by  \textbf{Theorem} \ref{thm_m22},
$\Pi$ is not context-hiding for $(f(\mathbf x),{\cal D}_f)=(\prod_{i=1}^d x_i,\mathbb F_p^d)$. 
For every polynomial $(g(\mathbf x),{\cal D}_g)=(\alpha(\prod_{i=1}^d (({\bf x}+{\bf c}){\bf l}_i^\top + e_i)+\beta)+\gamma, \mathbb F_p^d)$ with $(\alpha,\beta, \gamma, 
     ({\bf l}_i^\top)_{i=1}^d,{\bf c}, (e_i)_{i=1}^d) 
     \in 
     \mathbb F_p^*\times \mathbb F_p\times \mathbb F_p \times \mathbb{GL}_d(\mathbb F_p) \times 
     \mathbb F_p^d\times \mathbb F_p^d$, it holds that $(f(\mathbf x),{\cal D}_f)\equiv (g(\mathbf x),{\cal D}_g)$,  since there exists an equivalence transformation $S=(\alpha,\beta, \gamma, 
     ({\bf l}_i^\top)_{i=1}^d,{\bf c}, (e_i)_{i=1}^d)$ such that Eq. (\ref{eqn:ep1}) and   (\ref{eqn:ep2}) are satisfied.
By applying \textbf{Theorem} \ref{thm_equiv}, it is obtained that $\Pi$ is not context-hiding for $(g(\mathbf x),{\cal D}_g)$.

\section{Conclusion}
In this paper, we present a novel formalization of the context-hiding property for individual functions in IT-HSS and analyze this property of the Shamir-based IT-HSS scheme for monomials. By defining the equivalence of polynomials, we extend the results for monomials to equivalent polynomials. 
While this extension captures a wide range of equivalent polynomials, 
it remains an interesting problem to study for those polynomials outside this range.

\end{document}